\documentclass[12pt,a4paper]{amsart}
\usepackage{amsmath,amssymb,amsthm,booktabs,comment,longtable}
\usepackage{enumitem}
\usepackage{lscape}
\usepackage{graphicx}
\usepackage[all]{xy}
\usepackage{ccaption}
\usepackage{stmaryrd}
\usepackage{hyperref}
\usepackage{ascmac}
\usepackage{datetime2}
\usepackage{fancyhdr} 
\usepackage{mcommand}
\usepackage{xcolor}
\newtheorem{thm}{Theorem}[section]
\newtheorem{prop}[thm]{Proposition}
\newtheorem{lem}[thm]{Lemma}
\newtheorem{cor}[thm]{Corollary}
\newtheorem{dfn}[thm]{Definition}
\newtheorem{notation}[thm]{Notation}
\newtheorem{note}[thm]{Note}
\newtheorem{prob}[thm]{Problem}
\newtheorem{fact}[thm]{Fact}
\newtheorem{setting}[thm]{Setting}
\newtheorem{ex}[thm]{Example}
\newtheorem{rem}[thm]{Remark}

\newtheorem{conj}[thm]{Conjecture}

\newcommand{\A}{\forall}

\newcommand{\conv}{\operatorname{conv}}

\newcommand{\mcB}{{\mathcal B}}

\newcommand{\mcH}{{\mathcal H}}

\newcommand{\ep}{\varepsilon}

\begin{document}
\title[joint measurability of coplanar POVMs]{Joint measurability of coplanar POVMs}


\author[T.~Yoshino]{Taro Yoshino}
\address{Graduate School of Mathematical Sciences, The University of Tokyo\\ 3-8-1 Komaba, Meguro-ku, Tokyo 153-8914, Japan}
\email{yoshino@ms.u-tokyo.ac.jp}

\author[K.~Tojo]{Koichi Tojo}
\address{Department of Mathematical Sciences, Tokai University, 4-1-1 Kitakaname, Hiratsuka-shi, Kanagawa 259-1292, Japan}
\email{koichi.tojo@tokai.ac.jp}

\author[K.~Hagihara]{Kei Hagihara}
\address{Department of Mathematics, Faculty of Science and Technology, Keio University, 3-14-1 Hiyoshi, Kouhoku-ku, Yokohama 223-8522, Japan}
\email{kei.hagihara@gmail.com} 

\author[A.~Tanaka]{Akinori Tanaka}
\address{Mathematical Science Team, RIKEN Center for Advanced Intelligence Project (AIP), Nihonbashi 1-chome Mitsui Building, 15th floor, 1-4-1 Nihonbashi, Chuo-ku, Tokyo 103-0027, Japan}
\email{akinori.tanaka@riken.jp} 

\author[T.~Tomita]{Takuki Tomita}
\address{Mathematical Science Team, RIKEN Center for Advanced Intelligence Project (AIP), Nihonbashi 1-chome Mitsui Building, 15th floor, 1-4-1 Nihonbashi, Chuo-ku, Tokyo 103-0027, Japan}
\email{takuki.tomita@riken.jp}

\maketitle
\begin{abstract}
An unbiased qubit positive operator-valued measure (POVM) can be uniquely expressed in terms of the Bloch vector.
We prove that for POVMs with coplanar Bloch vectors, 
they are jointly measurable if and only if 
the perimeter of the convex closure of Bloch vectors is less than or equal to 4. 
Moreover, as a corollary, we give an affirmative answer to the conjecture suggested by [Andrejic--Kunjwal 2020] on the existence of a joint device for POVMs whose Bloch vectors lie on the same circle. 
\end{abstract}

\tableofcontents

\section{Introduction and main theorem}\label{section:introduction}
In quantum information theory, measurement compatibility is one of the central problems. 
In particular, it has been extensively studied for the case of binary qubit positive operator-valued measures (POVMs). 
An unbiased binary qubit POVM can be uniquely expressed in terms of its Bloch vector. 
Grinko and Uola \cite{GU25} characterized the compatibility of unbiased binary qubit POVMs as an optimization problem in terms of their corresponding Bloch vectors. 
Furthermore, using this characterization, they provided a counterexample to Conjecture~3 in \cite{AK20}. 
However, Conjecture~2 in the same paper \cite{AK20}, which concerns the case where the Bloch vectors lie on a common circle within a single plane, has remained open. 
In this paper, we affirmatively resolve this Conjecture~2.

\begin{setting}
Let $\mcH:=\C^2$, $\Omega:=\{\pm1\}$. 
Let $\mcB_+(\mcH)$ be the set of positive semi-definite operators on $\mcH$ and $\sigma=(\sigma_1,\sigma_2,\sigma_3)$ Pauli matrices. 
Any unbiased binary qubit POVM is uniquely determined by a vector $a\in \R^3$ with $\|a\|\le 1$ via the relation 
\[ E(\ep)=\frac12(I+\ep a\cdot \sigma) \qquad (\ep\in\Omega). \]
Here, $a\cdot \sigma:=a_1\sigma_1+a_2\sigma_2+a_3\sigma_3$. 
The vector $a$ is called the \textbf{Bloch vector} of $E$. 
Through this relation, we naturally identify an unbiased binary qubit POVM $E \colon \Omega \to \mcB_+(\mcH)$ with its Bloch vector $a \in \R^3$ ($\|a\| \leq 1$).
\end{setting}

\begin{dfn}
Let $E_1,\dots,E_n$ be POVMs on $\Omega$. 
A POVM $J\colon \Omega^n\to \mcB_+(\mcH)$ is  a \textbf{joint device} of $E_1,\dots,E_n$ if 
\[ \sum_{\omega\in \Omega^n\;:\; \omega_i=\varepsilon} J(\omega) = E_i(\varepsilon)\qquad (\A \varepsilon\in \Omega,\quad \A i\in\{1,\dots,n\}). \]
We say that $E_1,\dots,E_n$ are \textbf{jointly measurable} if they admit a joint device. 
\end{dfn}

\begin{thm}\label{thm:main} 
For coplanar unbiased binary qubit POVMs $a_1,\dots,a_n\in\R^2\subset\R^3$, the following conditions are equivalent. 
\begin{enumerate}
\item[\rm(i)] $a_1,\dots,a_n$ are jointly measurable. 
\item[\rm(ii)] The perimeter of the convex polygon $\conv\{ \pm a_i\ |\ i\in \{1,\dots,n\}\}$ is less than or equal to $4$. 
\end{enumerate}
Here, $\conv\{ \pm a_i\ |\ i\in \{1,\dots,n\}\}$ denotes the convex closure of the set $\{ \pm a_i\ |\ i\in \{1,\dots,n\}\}$. 
\end{thm}
The theorem follows directly from Lemma~\ref{lem:existence_of_joint_device_and_the_value_of_L} in Section~\ref{section:function_L_n} and Proposition~\ref{prop:2dim} in Section \ref{section:2-dim_case}.

\begin{rem}
The condition $\|a+b\|+\|a-b\|\le 2$ for two unbiased binary qubit POVMs $a,b\in \R^3$ to be jointly measurable (\cite{B86}) is a special case of Theorem~\ref{thm:main}. 
\end{rem}

Although the implication (ii)$\Rightarrow$(i) of Conjecture~\ref{conj:ak} is proved by [AK20], the converse implication has been open.

\begin{conj}[{\cite[Conjecture~2]{AK20}}]\label{conj:ak}
Let $0=\alpha_0 < \alpha_1 <\alpha_2 < \dots < \alpha_{n-1} < \pi $ and $\eta\in(0,1]$. 
We define $a_i\in\R^3$ by 
\[ a_i=\eta \begin{pmatrix}\cos \alpha_{i-1}\\ \sin \alpha_{i-1}\\ 0\end{pmatrix} \in\R^3 \]
Then the following conditions are equivalent. 
\begin{enumerate}
\item[\rm(i)] $a_1,\dots,a_n$ are jointly measurable. 
\item[\rm(ii)] The following inequality holds.  
\[ \eta \le \frac 1 { \sum_{i=1}^{n-1}\sin \frac{\alpha_i-\alpha_{i-1}}2 + \cos\frac{\alpha_{n-1}}2}. \]
\end{enumerate}
\end{conj}

\begin{prop}\label{solution_to_conj_in_AK20}
Conjecture~\ref{conj:ak} is true. 
\end{prop}

\begin{rem}
In \cite{AK20}, the truth of Conjecture 2 implies the truth of Conjecture 1. 
Therefore, Conjecture 1 in \cite{AK20} is also true. 
\end{rem}

\begin{proof}[Proof of Proposition~\ref{solution_to_conj_in_AK20}]
The perimeter of the convex polygon $\conv\{\pm a_i \mid i\in \{1,\cdots, n\}\}$ is $4\eta (\sum_{i=1}^{n-1}\sin \frac{\alpha_i-\alpha_{i-1}}{2}+\cos \frac{\alpha_{n-1}}{2})$. 
Thus by Theorem~\ref{thm:main}, the inequality in (ii) holds. 
\end{proof}

\section{Preliminary}
First, we consider the general case without assuming $a_i \in \R^2$. 
That is, for $n \in \Z_{>0}$, we investigate the conditions for the existence of a joint device for $n$ devices $a_1, \dots, a_n \in \R^3$. 
The discussion proceeds as follows.

\par\medskip\noindent\underline{\bf Step~1.} Define a function $L_n(a_1,\dots,a_n)$ and show that the existence of a joint device is equivalent to the condition $L_n \le 1$. 

\par\medskip\noindent\underline{\bf Step~2.} We prove that the value of the function $L_n(a_1,\dots,a_n)$ depends only on the convex closure $\conv\{\pm a_i\ |\ i\in\{1,\dots,n\}\}$. 

\par\medskip\noindent\underline{\bf Step~3.} 
We verify that the value of the function $L_n(a_1,\dots,a_n)$ can be expressed as the minimum of a convex function $f$ subject to linear constraints.
Here $f$ is fixed and  the linear constraints depend on $a_i$. 

\par\medskip\noindent\underline{\bf Step~4.} 
Given a minimization problem of a convex function $f$ subject to linear constraints, we provide a necessary and sufficient condition for an element $x$ to be a minimizer of $f$.

\par\medskip\noindent\underline{\bf Step~5.} We explicitly state the condition derived in Step~4 under our present setting. 

\subsection{Definition and Basic Property of the function $L_n$}\label{section:function_L_n}

\begin{notation} We use the following notation. 
\begin{enumerate}
\item[$\bullet$] $\mcH:=\C^2$
\item[$\bullet$] $\Omega^n:=\{\pm1\}^n$
\item[$\bullet$] $D^n:=\{x\colon \Omega^n\to\R^3\ |\ \A \omega\in \Omega^n,\ x(-\omega)=-x(\omega)\}$
\item[$\bullet$] $\Omega_i^n:=\{\omega\in \Omega^n\ |\ \omega_i=1\}$ (for $i=1,\dots,n$).
\item[$\bullet$] For $x\in D^n$,
\begin{align*}
x_{(i)} &:= \sum_{\omega\in \Omega_i^n} x(\omega)\quad \in\R^3, \\
|x| &:=\sum_{\omega\in \Omega_1^n} \|x(\omega)\|. 
\end{align*}
\end{enumerate}
\end{notation}

In the definition of $|x|$, we may replace $\Omega_1$ with $\Omega_i$. In fact, the following holds:
\begin{note} 
For any $i\in\{1,\dots,n\}$ and $x\in D^n$, we have:
\[ |x| = \sum_{\omega \in \Omega_i^n} \|x(\omega)\|. \]
\end{note}

\begin{note} Let $E_1,\dots,E_n$ be unbiased binary qubit POVMs on $\Omega$. 
Then, 
a POVM $J\colon \Omega^n\to\mcB_+(\mcH)$ is a joint device of $E_1,\dots,E_n$ if and only if 
\[ \sum_{\omega\in \Omega^n_i} J(\omega) = E_i(1)\qquad (\text{for }i=1,\dots,n).  \]
\end{note}


\begin{dfn} We define a function $L_n$ on $(\R^3)^n$ as follows. 
\[ L_n(a_1,\dots,a_n) := \min\{ |x|\in \R_{\ge0} \ |\ x\in D^n \text{ such that } \A i, a_i = x_{(i)}\} \]
We may simply write $L$ for $L_n$ when $n$ is clear from the context.
\end{dfn}

\begin{rem}
Since the function $|\cdot | \colon D^n\to \R$ is convex, 
the value $L_n(a_1,\cdots, a_n)$ exists in $\R$ and is finite. 
\end{rem}

\begin{lem}\label{lem:existence_of_joint_device_and_the_value_of_L}
For unbiased binary qubit POVMs $a_1,\dots,a_n\in\R^3$ with $\|a_i\|\le1$, 
the following conditions are equivalent. 
\begin{enumerate}
\item[\rm(i)] $a_1,\dots,a_n$ is jointly measurable. 
\item[\rm(ii)] $L(a_1,\dots,a_n)\le 1$. 
\end{enumerate}
\end{lem}

To prove Lemma~\ref{lem:existence_of_joint_device_and_the_value_of_L}, we introduce the concept of the dual of a POVM and establish a preliminary lemma (Lemma~\ref{lem:dual}):

\begin{dfn}
Any $A \in \mcB_+(\mcH)$ can be uniquely expressed in terms of $t \in \R_{\ge 0}$ and $a \in \R^3$ as 
\[ A = t I + a \cdot \sigma,\] where $\|a\| \leq t$. 
We then define
\[A^* := t I - a \cdot \sigma \in \mcB_+(\mcH)\] 
and call it the \textbf{Pauli dual of $A$}.
For a POVM $J\colon \Omega^n\to \mcB_+(\mcH)$, we define its \textbf{dual} $J^*$ by 
\[ J^*(\omega):=(J(-\omega))^*. \]
We call $J$ \textbf{self dual} if $J=J^*$. 
\end{dfn}

\begin{lem}\label{lem:dual}
For unbiased binary qubit POVMs $E_1,\dots,E_n$, the following conditions are equivalent. 
\begin{enumerate}
\item[\rm(i)] $E_1,\dots,E_n$ admit a joint device. 
\item[\rm(ii)] $E_1,\dots,E_n$ admit a self dual joint device. 
\end{enumerate}
\end{lem}
\begin{proof}
Since (ii)$\Rightarrow$(i) is clear, we verify (i)$\Rightarrow$(ii).
It is straightforward to see the following: 
\par\medskip\noindent\underline{\bf Claim.} 
\begin{enumerate}
\item[(1)] If $J$ is a joint device of $E_1,\dots,E_n$, then so is $J^*$. 
\item[(2)] If $J_1$, $J_2$ are joint devices of $E_1,\dots,E_n$, then so is $(J_1+J_2)/2$.
\end{enumerate}
In fact, (1) follows from 
\[ \sum_{\omega\in\Omega^n\;:\;\omega_i=\ep} J^*(\omega) = \left(\sum_{\omega\in\Omega^n\;:\;\omega_i=-\ep} J(\omega)\right)^* 
                                     = (E_i(-\ep))^*      = E_i(\ep). 
\]
(2) is also straightforward. 

Now assume (i) and let $J$ be a joint device of $E_1,\dots,E_n$. 
By Claim, $(J+J^*)/2$ is a self dual joint device of 
$E_1,\dots,E_n$, which proves (ii). 
\end{proof}

\begin{proof}[Proof of Lemma~\ref{lem:existence_of_joint_device_and_the_value_of_L}]
It is enough to show that the following conditions are equivalent. 
\begin{enumerate}
\item[\rm(i)] $E_1,\dots,E_n$ admit a joint device. 
\item[\rm(ii)] There exists a map $x\colon \Omega^n\to\R^3$ such that 
\begin{enumerate}
\item[\rm(a)] $x(-\omega)=-x(\omega)$ for any $\omega\in \Omega^n$, 
\item[\rm(b)] $|x|\leq 1$, 
\item[\rm(c)] $x_{(i)}=a_i$ ($i=1, \cdots, n$).
\end{enumerate}
\end{enumerate}
First we show (i)$\Rightarrow$(ii). 
Let $J$ be a joint device of $E_1,\cdots, E_n$. 
By Lemma~\ref{lem:dual}, we may assume $J$ is self dual. 
Here, $J$ can be written using some maps $\tau\colon \Omega^n\to\R_{\ge0}$ and $\xi\colon \Omega^n\to\R^3$ as 
\[ J(\omega) = \frac12(\tau(\omega) I + \xi(\omega)\cdot \sigma)\qquad (\omega\in\Omega^n). \]
Then, one can easily check the following: 
\par\noindent\underline{\bf Claim.}
\begin{enumerate}
\item[(1)] For any $\omega\in \Omega^n$, we have
$\|\xi(\omega)\| \le \tau(\omega)$, $\tau(\omega) = \tau(-\omega)$, and $\xi(-\omega)=-\xi(\omega)$.
\item[(2)] For any $i=1,\dots,n$, we have $\ds \sum_{\omega\in \Omega^n_i} \tau(\omega) = 1$ and $\ds \sum_{\omega\in \Omega^n_i} \xi(\omega) = a_i$
\end{enumerate}
In particular, the map $\xi\colon\Omega^n\to\R^3$ satisfies the conditions (a), (b) and (c). 
\par\medskip We show (ii)$\Rightarrow$(i). Take a map $\xi\colon\Omega^n\to\R^3$ satisfying the conditions (a), (b) and (c). 
By the condition (b), we can and do take a map $\tau_1\colon \Omega^n_1\to\R_{\ge0}$ such that 
\[ \A \omega\in \Omega^n_1,\ \|\xi(\omega)\|\le \tau_1(\omega) \]
and 
\[ \sum_{\omega\in \Omega^n_1} \tau_1(\omega) = 1. \]
Here, we note that for any $\omega \in \Omega^n$, exactly one of $\omega \in \Omega^n_1$ and $-\omega \in \Omega^n_1$ holds, 
and define $\tau\colon \Omega^n\to\R_{\ge0}$ by 
\[ \tau(\omega):= \begin{cases}
\tau_1(\omega) &\text{(if $\omega\in \Omega^n_1$)},\\
\tau_1(-\omega) &\text{(if $-\omega \in \Omega^n_1$)}.\\
\end{cases} \]
Furthermore, we define $J\colon \Omega^n\to \mcB_+(\mcH)$ by 
\[ J(\omega) := \frac12(\tau(\omega)I+\xi(\omega)\cdot\sigma). \]
Then, one can easily verify that $J$ is a joint device of $E_1,\dots,E_n$. 
\end{proof}

\begin{lem}\label{lem:conv}
Let $a_1,\dots,a_n\in\R^3$. 
The value of the function $L(a_1,\dots,a_n)$ depends only on 
the convex set $\conv\{\pm a_i\ |\ i\in\{1,\dots,n\}\}$. 
\end{lem}
\begin{proof}
It is enough to show the following: 
\par\noindent\underline{\bf Claim.}
\begin{enumerate}
\item[\rm(i)] For any permutation $\sigma\in S_n$ of degree $n$, 
the following equality holds:
\[ L(a_1,\dots,a_n)=L(a_{\sigma(1)},\dots,a_{\sigma(n)}). \]
\item[\rm(ii)] The following equality holds:  
\[ L(a_1,\dots,a_n)=L(a_1,\dots,a_{n-1},-a_n).  \]
\item[\rm(iii)] For any $a_{n+1}\in \conv\{\pm a_i\ |\ i=1,\dots,n\}$, the following equality holds:
\[ L_n(a_1,\dots,a_n)=L_{n+1}(a_1,\dots,a_{n+1}). \]
\end{enumerate}

Since (i) is clear and (ii) is easy to verify, 
we show only (iii). 
Take any $a_{n+1}\in \conv\{\pm a_i\ |\ i=1,\dots,n\}$. 

First, we show the following:
\par\noindent\underline{\bf Subclaim.} $L_n(a_1,\dots,a_n)\le L_{n+1}(a_1,\dots,a_{n+1})$.\\
Take $x\in D^{n+1}$ such that
\begin{enumerate}
\item[\rm(i)] $|x|=L_{n+1}(a_1,\dots,a_{n+1})$, 
\item[\rm(ii)] $a_i = x_{(i)}\qquad (i=1,\dots,n+1)$, 
\end{enumerate}
and define $y\colon \Omega^n\to\R^3$ by 
\[ y(\omega):=x(\omega,1)+x(\omega,-1)\ \in \R^3\qquad (\omega\in \Omega^n). \]
It is enough to show the following conditions, which can be easily verified. 
\begin{enumerate}
\item[$\bullet$] $y\in D^n$, 
\item[$\bullet$] $|y|=|x|$, 
\item[$\bullet$] $\A i\in \{1,\cdots,n\},\ a_i=y_{(i)}$. 
\end{enumerate}

Next, we show the following:
\par\noindent\underline{\bf Subclaim.} $L_n(a_1,\dots,a_n) \ge L_{n+1}(a_1,\dots,a_{n+1})$.\\
Take $x\in D^n$ such that 
\begin{enumerate}
\item[\rm(i)] $|x|=L_n(a_1,\dots,a_n)$, 
\item[\rm(ii)] $a_i = x_{(i)}\qquad (i=1,\dots,n)$. 
\end{enumerate}
Take 
$c_1,\dots,c_n\in\R$ such that 
\begin{enumerate}
\item[\rm(i)] $a_{n+1}=\sum_{i=1}^n c_i a_i$
\item[\rm(ii)] $\sum_{i=1}^n|c_i|\le 1$
\end{enumerate}
For $\omega\in \Omega^n$, we put 
\[ u_\omega :=\sum_{i=1}^n \omega_i c_i \in \R,  \]
and define a map $y\colon \Omega^{n+1}\to\R^3$ by 
\[ y(\omega,\pm 1):=\frac{1\pm u_\omega}2 x(\omega). \]
Then, it is enough to show 
\par\noindent\underline{\bf Subsubclaim.}
\begin{enumerate}
\item[\rm(i)] $y\in D^{n+1}$, 
\item[\rm(ii)] $a_i=y_{(i)}\qquad (i=1,\dots,n+1)$, 
\item[\rm(iii)] $|y|=|x|$. 
\end{enumerate}

This Subsubclaim comes from the following: 
\begin{enumerate}
\item[\rm(i)] $|u_\omega|\le 1$, 
\item[\rm(ii)] $u_{-\omega} = -u_\omega$, 
\item[\rm(iii)] $\sum_{\omega\in \Omega^n} x(\omega)=0$, 
\item[\rm(iv)] $x_{(i)} = \frac12\sum_{\omega\in \Omega^n} \omega_i x(\omega)$. 
\end{enumerate}

\end{proof}

\begin{cor} The existence of a joint device of $a_1,\dots,a_n\in \R^3$ depends only on the convex polytope $\conv\{\pm a_i\ |\ i=1,\dots,n\}$. 
\end{cor}

\subsection{Matrix formulation of the minimization problem}\label{section:matrix_formulation}

We use the notation $D^n$, $|x|$, $x_{(i)}$ as in the previous section. 
By the definition of the function $L_n$, determining the value of $L_n$ is equivalent to solving the following minimization problem.

\begin{prob}\label{prob:min}
Minimize $|x|$ for $x\in D^n$ subject to $\A i$, $a_i=x_{(i)}$.
\end{prob}

By rewriting this problem in matrix form, we aim to formulate it as a minimization problem of a convex function subject to linear constraints.
\begin{notation}
Let $N := 2^{n-1}$. For each $i \in \{0, \dots, N-1\}$, we represent $i$ as an $n$-bit binary number. Replacing $0$ with $1$ and $1$ with $-1$, we regard the resulting sequence as an element of $\Omega^n$. We denote this element by $\beta_i \in \Omega^n$ and call it a binary vector.
\end{notation}
\begin{ex} 
For example, for $n=4$, we have:
\begin{align*} \beta_0&=(1,1,1,1)\in \Omega^4, \\
   \beta_1&=(1,1,1,-1)\in \Omega^4, \\
\vdots \\
   \beta_7&=(1,-1,-1,-1)\in \Omega^4. \end{align*}
\end{ex}

Noting that $x \in D^n$ is uniquely determined by its values on $\beta_0, \dots, \beta_{N-1}$, we naturally identify $D^n$ with $M(3,N)$ via the correspondence
\[x \mapsto [x(\beta_0),\dots,x(\beta_{N-1}) ] \in M(3,N).\] 
Let $B$ be the $N \times n$ matrix whose rows are given by $\beta_0, \dots, \beta_{N-1}$. Namely, we put 
\begin{equation} \label{eq:def:B}
 B:=\begin{pmatrix}\beta_0\\\vdots\\\beta_{N-1}\end{pmatrix} \in M(N,n). 
\end{equation}

Moreover, we set 
\begin{equation}\label{eq:def:A}
A:=(a_1,\dots,a_n) \in M(3,n). 
\end{equation}

\begin{note} 
Via the identification $x\in D^n \simeq M(3,N)$, 
we obtain 
\[ \A i, a_i = x_{(i)} \iff A=x B. \]
\end{note}

Furthermore we define a convex function $f\colon M(3,N)\to\R$ by \begin{equation}\label{eq:def:f}
f(x):=\sum_{j=1}^N \|x_j\| \qquad (x\in M(3,N)), 
\end{equation}
where $x_j\in \R^3$ is the $j$-th column vector of $x$.

Then, Problem~\ref{prob:min} can be reformulated as follows: 
\begin{prob}\label{prob:fmin}
Minimize $f(x)$ for $x\in M(3,N)$ subject to $A=x B$.
\end{prob}

\subsection{Minimizing problem with linear constraints}
In this section, we review the general theory of the minimization problem of a convex function subject to linear constraints. 
\begin{setting} 
Let the following be given:
\begin{enumerate}
\item[$\bullet$] $V,W$: finite dimensional real vector spaces
\item[$\bullet$] $f\colon V\to\R$: a convex function
\item[$\bullet$] $g\colon V\to W$: a linear map 
\item[$\bullet$] $w_0\in W$
\end{enumerate}
\end{setting}
We consider the following minimization problem. 
\begin{prob}\label{prob:general}
Minimize $f(x)$ for $x\in V$ subject to $g(x)=w_0$.
\end{prob}

As a solution to this problem, we have the following well-known result.

\begin{fact} Assume that $x^*\in V$ satisfies $g(x^*)=w_0$. 
Then the following conditions are equivalent. 
\begin{enumerate}
\item[\rm(i)] $x=x^*$ is a minimizer for Problem~\ref{prob:general}.
\item[\rm(ii)] There exists $y\in W^\vee \text{ such that } g^\vee(y) \in \partial f(x^*)$. 
\end{enumerate}
Here, $V^\vee$, $W^\vee$ are the dual spaces of $V,W$, respectively, and $g^\vee\colon W^\vee\to V^\vee$ the dual map of $g$.
For $x\in V$, $\partial f(x)\subset V^\vee$ is the subdifferential of $f$ at $x$.
\end{fact}

\begin{proof}
We set $C=g^{-1}(w_0)$. Then, by \cite[Theorem~27.4]{R70}, the condition (i) is equivalent to $\partial f(x^\ast)\cap (-N_{C}(x^\ast))\neq \emptyset$, where $N_{C}(x^\ast)$ is the normal cone, that is,
\[
N_C(x^\ast) = \{\xi \in V^\vee \mid \forall y\in C, \,\xi(y-x^\ast)\le 0 \}.
\]

Since we easily check that
\begin{align*}
N_C(x^\ast) &= \{\xi\in V^\vee \mid \forall z\in \Ker g,\,\xi(z)\le 0 \} \\
 &= \{\xi\in V^\vee \mid \forall z\in \Ker g,\,\xi(z)= 0 \}  = g^\vee(W^\vee),
\end{align*}
 we obtain the desired equivalence.
\end{proof}

\subsection{Application to our setting}
We now apply the above general theory to our setting. 
Define $A\in M(3,n),B\in M(N,n)$ and a convex function $f$ as \eqref{eq:def:A}, \eqref{eq:def:B} and \eqref{eq:def:f} in Section~\ref{section:matrix_formulation}, respectively. 

We put 
\begin{align*} V&:=M(3,N), \\
   W&:=M(3,n), \\
   w_0&:=A\in W,  \\
   g(x)&:=x B. 
\end{align*}
We identify $V^\vee,W^\vee$ with $V, W$, respectively via the Frobenius inner product ($\langle A,B\rangle = \tr (\t\!A B)$). 

A direct calculation yields the following: 
\begin{lem} Let $x\in V$.
\[ \partial f(x) = \{ (v_1,\dots, v_N) \in M(3,N)\ |\ \text{The following conditions (i) and (ii) hold} \} \]
\begin{enumerate}
\item[\rm(i)] $\A j\in J,\ v_j = \frac{x_j}{\|x_j\|}$. 
\item[\rm(ii)] $\A j\in\{1,\dots,N\},\ \|v_j\|\le 1$. 
\end{enumerate}
Here, 
\[ J:=\{j\in\{1,\dots,N\}\ |\ x_j\ne0\}. \]
\end{lem}

\begin{lem} The dual map $g^\vee$ of $g$ is given as follows:
\[ g^\vee(y) = y\;\t\! B\qquad (y\in M(3,n)). \]
\end{lem}
\begin{proof}
Take any $x\in V = M(3,N)$ and $y\in W^\vee=M(3,n)$. 
The following calculation shows our assertion. 
\[ \langle x, g^\vee(y)\rangle = \langle x B, y\rangle = \tr\left(\t(x B)y\right) = \tr\left(\t x y\t\!B\right)
=\langle x, y\t\!B\rangle
 \]
\end{proof}

Summarizing the above discussion, we obtain the following.
\begin{lem}\label{lem:lvalue} Let $c\in\R$. 
The following conditions are equivalent. 
\begin{enumerate}
\item[\rm(i)] $L(a_1,\dots,a_n)=c$. 
\item[\rm(ii)] There exist $X=(x_1,\dots,x_N) \in M(3,N)$ and $Y\in M(3,n)$ such that
\begin{enumerate}
\item\label{item:norm_condition} $\sum_{j\in J} \|x_j\|=c$, 
\item\label{item:linear_restriction} $A=XB$, 
\item\label{item:subderivative_1} $\A j\in J,\ (Y \;\t\!B)_j=\hat x_j$, 
\item\label{item:subderivative_2} $\A j\in\{1,\dots,N\},\ \|(Y\;\t\!B)_j\|\le 1$, 
\end{enumerate}
\end{enumerate}
where
\begin{align*} J&:=\{j\in\{1,\dots,n\}\ |\ x_j\ne 0\}, \\
   \hat x_j &:= \frac{x_j}{\|x_j\|} \qquad (\text{for $j\in J$}). 
\end{align*}
\end{lem}

\bigskip
Namely, to show that the value of the function $L$ is equal to a specific $c \in \R$, it suffices to find matrices $X$ and $Y$ satisfying conditions (\ref{item:norm_condition}), (\ref{item:linear_restriction}), (\ref{item:subderivative_1}), and (\ref{item:subderivative_2}).

\section{2-dim case}\label{section:2-dim_case}
Let $a_1,\dots,a_n\in\R^2\subset \R^3$. 
Our goal of this section is to prove the following:
\begin{prop}\label{prop:2dim}
The value of the function $L(a_1,\dots,a_n)$ is equal to one-fourth of the perimeter of the convex polygon $\conv\{\pm a_i\ |\ i\in\{1,\dots,n\}\}$.
\end{prop}
Our main result (Theorem~\ref{thm:main}) directly follows from this proposition combining with Lemma~\ref{lem:existence_of_joint_device_and_the_value_of_L}. 

The discussion proceeds as follows.
\par\medskip\noindent\underline{\bf Step~1.} We introduce a condition called well-arranged. We verify that it suffices to prove the proposition only for the case where $a_1, \dots, a_n$ are well-arranged.

\par\medskip\noindent\underline{\bf Step~2.} We construct the matrices $X$ and $Y$ via auxiliary matrices $K, \hat{K}, P,$ and $C$.

\par\medskip\noindent\underline{\bf Step~3.} We show that $X$ and $Y$ satisfy conditions (\ref{item:norm_condition}), (\ref{item:linear_restriction}), and (\ref{item:subderivative_1}) through a relatively straightforward direct calculation.

\par\medskip\noindent\underline{\bf Step~4.} Condition (\ref{item:subderivative_2}) requires that the norm of a certain vector be at most $1$. To this end, we prove a necessary norm inequality, which we then use to establish condition (\ref{item:subderivative_2}). 

\subsection{Well-arranged}
We introduce a condition called well-arranged and verify that it is enough to prove Proposition~\ref{prop:2dim} only for the case where $a_1, \dots, a_n$ are well-arranged.

\begin{dfn} We say $a_1,\dots,a_n\in \R^2$ are {\bf well-arranged} if 
\begin{enumerate}
\item[\rm(i)] $a_1$,\dots, $a_n$ are distinct from each other and nonzero. 
\item[\rm(ii)] $\arg a_1 < \arg a_2 < \dots <\arg a_n < \arg a_1 + \pi (= \arg a_{n+1})$
\item[\rm(iii)] $\arg(a_2-a_1) < \arg(a_3-a_2) < \dots < \arg (a_{n+1}-a_n)$
\end{enumerate}
Here, $a_{n+1}:=-a_1$.
\end{dfn}

By Lemma~\ref{lem:conv}, the value of $L$ depends only on the convex closure $\{\pm a_i\ |\ i\in\{1,\dots,n\}\}$.
By removing unnecessary $a_i$, rearranging them appropriately, and replacing $a_i$ with $-a_i$ if necessary, 
we can make $a'_1, \dots, a'_{n'}$ well-arranged without changing the convex hull $\{\pm a_i \mid i \in \{1, \dots, n\}\}.$
Therefore, it suffices to prove Proposition~\ref{prop:2dim} for the case where $a_1, \dots, a_n$ are well-arranged. 

\subsection{Construction of $X$ and $Y$}
Now, we assume that $a_1,\dots,a_n \in \R^2\subset \R^3$ are well-arranged.

Below, we construct the matrices $X$ and $Y$ of Lemma~\ref{lem:lvalue} explicitly and determine the value of $L$. 
We first list the sizes of the matrices that appear throughout the discussion.
\begin{center}
\begin{tabular}{c|c}
$(3,n)$ & $A,K,\hat K,Y$ \\
$(N,n)$ & $B,P$ \\
$(3,N)$ & $X$ \\
$(n,n)$ & $C$ 
\end{tabular}
\end{center}
Here, $A$ and $B$ are the matrices defined previously in Section~\ref{section:matrix_formulation}.

\begin{dfn} We define matrices $K,\hat K$. 
First, we put 
\[ k_i:= a_i - a_{i+1} \qquad (i\in\{1,\dots,n\}). \]
Since $\{a_i\}$ is well-arranged, $k_i\ne0$ for any  $i\in \{1,\dots,n\}$. Then we put
\[ \hat k_i:=\frac{k_i}{\|k_i\|} \]
and define $K,\hat K$ by 
\begin{align*} K&:=(k_1,\dots,k_n) \in M(3,n), \\
   \hat K&:=(\hat k_1,\dots,\hat k_n) \in M(3,n). 
\end{align*}
\end{dfn}

\begin{picture}(100,0)(-260,20)
\put(0,0){\includegraphics{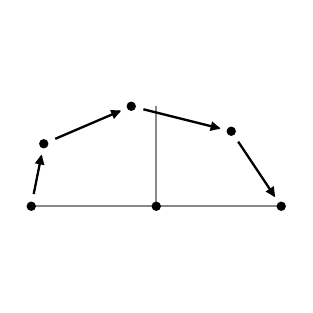}}
\put(145,50){$a_1$}
\put(117,92){$a_2$}
\put(60,110){$a_3$}
\put(5,85){$a_4$}
\put(-10,50){$-a_1$}
\put(110,65){$k_1$}
\put(80,85){$k_2$}
\put(40,80){$k_3$}
\put(20,60){$k_4$}
\end{picture}

\begin{dfn} We define matrices $P$ and $C$. For $i=1,\dots,n$, 
we use the following notation:
\[ [i]:=2^{n-i}. \]
We define $N\times n$ matrix $P$ and $n\times n$ matrix $C=(c_{ij})$ as follows:
\begin{align*} 
(P)_j &:= e_{[j]}\\
c_{ij} &:= \begin{cases}
1 &\text{($i=j$ or $(i,j)=(1,n)$)}\\
-1 &\text{($i=j+1$)}\\
0 &\text{(otherwise)}\\
\end{cases}
\end{align*}
Here, $(P)_j$ denotes the $j$-th column vector of $P$, and $e_1,\dots,e_N$ denote the standard basis vectors of $\R^N$. 
\end{dfn}
\begin{note}
The equality $K=A C$ holds. 
\end{note}
\begin{ex} For $n=4$,
\[ C=\begin{pmatrix}1&0&0&1\\
              -1&1&0&0\\
              0&-1&1&0\\
              0&0&-1&1\end{pmatrix}\in M(4),\qquad P=(e_8,e_4,e_2,e_1) \in M(8,4). \]
\end{ex}
\begin{dfn} We define matrices $X,Y$ as follows:
\begin{align*} X&:= \frac12 K\;\t\!P \in M(3,N),\\
   Y&:= \frac12 \hat K\;\t\!C \in M(3,n).\end{align*}
\end{dfn}

To prove Proposition~\ref{prop:2dim}, it is enough to show the following conditions: 


\begin{enumerate}
\item[\rm(a)] $\sum_{j\in J} \|x_j\|=\frac{1}{4}(\text{the perimeter of }\conv\{\pm a_i\ |\ i\in\{1,\dots,n\}\})$, 
\item[\rm(b)] $A=XB$, 
\item[\rm(c)] $\A j\in J, (Y\;\t\!B)_j=\hat x_j$, 
\item[\rm(d)] $\A j\in\{1,\dots,N\},\ \|(Y\;\t\!B)_j\|\le 1$, 
\end{enumerate}
where
\begin{align*} J&:=\{j\in\{1,\dots,n\}\ |\ x_j\ne 0\}, \\
   \hat x_j &:= \frac{x_j}{\|x_j\|} \qquad (\text{for $j\in J$}). 
\end{align*}

In what follows, we verify these conditions.
\subsection{The condition (a)}
We determine the set $J$ and verify the condition (a).

\begin{ex}Let $\varepsilon_i$ be the standard basis of $\R^n$. For $n=4$, since we have 
\[ \t\!P=(\varepsilon_4,\varepsilon_3,0,\varepsilon_2,0,0,0,\varepsilon_1) \in M(4,8), \]
we get 
\begin{align*} X &= \frac12 K\t\!P \\
     &= \frac12 (k_4,k_3,0,k_2,0,0,0,k_1). \end{align*}
Therefore, 
\[ J=\{1,2,4,8\} \]
holds, and furthermore, 
\[ \sum_{j\in J} \|x_j\| = \frac 12\sum_{i=1}^4 \|k_i\| = \frac14(\text{the perimeter}) \]
holds. 
\end{ex}

Similarly, in general case, 
since we have 
\[ (\t\!P)_j= \begin{cases}
\varepsilon_i &\text{($[i]=j$)},\\
0   &\text{(otherwise)},\\
\end{cases} \]
we get 
\[ x_j = \begin{cases}
\frac12k_i &\text{($[i]=j$)},\\
0  &\text{(otherwise)}.\\
\end{cases} \]
Therefore, 
\[ J=\{[i]\ |\ i\in\{1,\dots,n\}\} \]
holds, and furthermore, 
\[\sum_{j\in J}\|x_j\| =\frac12\sum_{i=1}^n \|k_i\|=\frac14(\text{the perimeter}) \]
holds. 
\subsection{The conditions (b) and (c)}
We verify the conditions (b) and (c). To this end, we use the following: 
\begin{lem}\label{lem:relation_BCP}
\[ C\inv = \frac12 \;\t\!P B. \]
\end{lem}

\begin{proof}
A direct calculation implies that 
\begin{align*}
(\text{the $i$-th row vector of } \t\! PB)= \beta_{[i]-1}\quad (i=1,\cdots,n). 
\end{align*}
Furthermore, by the definition of $C$, we get 
\begin{align*}
\text{(the $i$-th row vector of $C\t\! P B$)} =
\begin{cases}
 \beta_{[1]-1} + \beta_{[n]-1}=2 \t\!\varepsilon_1 & (i=1),\\
 \beta_{[i]-1} - \beta_{[i-1]-1}=2\t\!\varepsilon_i & (2\le i\le n),
\end{cases}
\end{align*}
where $\varepsilon_1,\cdots, \varepsilon_n$ denote the standard basis vectors of $\R^n$. 
Thus, we obtain $C\t\! PB=2I_n$. 
\end{proof}

This lemma leads to the condition (b) ($A = XB$) through the following direct computation.
\[ X B = \frac 1 2 K \;\t\!P B
       = \frac 1 2 A C \;\t\!P B
       = A.  \]
Next, we verify the condition (c) ($\A j\in J, (Y \;\t\!B)_j=\hat x_j$). 
Take any $j\in J$. Then we can and do take $i\in\{1,\dots,n\}$ such that $[i]=j$. 
Noting that $P \varepsilon_i = e_{[i]}$ holds, 
we have 
\begin{align*} (Y\;\t\!B)_j & = Y\;\t\! B e_{[i]} \\
                & = \frac12 \hat K \;\t\! C \;\t\! B P \varepsilon_i \\
                & = \frac12 \hat K \t\!(\t\!P B C) \varepsilon_i \\
                & = \hat K \varepsilon_i \\
                & = \hat k_i. 
\end{align*}
Here, we used Lemma~\ref{lem:relation_BCP} for the third equality above. 
On the other hand, $x_j = x_{[i]} = \frac12 k_i$ holds. Thus, we obtain $\hat x_j = \hat k_i$, which proves that the condition (c) is satisfied.

\subsection{Norm inequality}
All that remains is to verify condition (d).
This asserts that the norm of a certain vector is at most $1$.
To prove this, we need the following: 

\begin{prop}\label{prop:norm_ineq} Let $n\in\Z_{\ge0}$. 
Assume that the $(2n+1)$ real numbers $\theta_0, \dots, \theta_{2n} \in \R$ satisfy the inequalities $\theta_0 \le \dots \le \theta_{2n} \le \theta_0+\pi$. 
Then, the following inequality holds:
\[ \left\|\sum_{i=0}^{2n}(-1)^i e(\theta_i) \right\| \le 1, \]
where, for $\theta\in\R$, $e(\theta)\in\R^2$ is defined by 
\[ e(\theta) := \begin{pmatrix}\cos\theta\\\sin\theta\end{pmatrix} \in \R^2. \]
\end{prop}

To prove this proposition, for $s,t\in\R$, we define a closed subset $D(s,t)\subset \R^2$ as follows:
\[ D(s,t):=\{x\in\R^2 \ |\ \|x\|\le 1,\ \|x-(e(s)+e(t))\|\le 1 \}. \]
\begin{picture}(150,10)(-320,110)
\put(0,0){\includegraphics{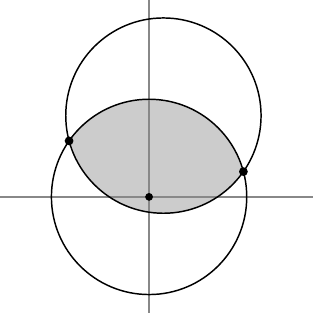}}
\put(123,65){$e_s$}
\put(17,82){$e_t$}
\put(60,74){$D(s,t)$}
\end{picture}

This subset has the following property. 
\begin{lem}\label{note:dst} For $s\le s'\le t' \le t\le s+\pi$, we have 
\[ D(s',t')\subset D(s,t). \]
\end{lem}
\begin{proof}
We show that $s=s'$ and $s\le t' \le t \le s+\pi$ imply
\[ D(s,t')\subset D(s,t). \]
Once this is shown, a similar argument yields $D(s',t') \subset D(s,t')$. 
Furthermore, by the rotational symmetry of the problem, it suffices to prove the case where $s=0$.
Namely, we show that when $0\le t' \le t \le \pi$, 
\[ D(0,t')\subset D(0,t) \]
holds. We put
\[ E(t):=\{ x\in\R^2\ |\ \|x-e(t)\|\le 1\}. \]
The region obtained by translating $D(0,t)$ by $-1$ along the $x$-axis is given as follows:
\begin{align*} D(0,t)-\begin{pmatrix}1\\0\end{pmatrix} & = \{ x\in\R^2\ |\ \|x+e(0)\|\le1,\ \|x-e(t)\|\le1\} \\
                       & = E(t)\cap E(\pi). \end{align*}
Therefore, it is enough to show
\[ E(\pi)\cap E(t') \subset E(\pi)\cap E(t).\]
Take any $x\in E(\pi)$. 
Then we can write $x$ as $x = r e(\theta)$, where $r \in [0,2]$ and $\theta \in [\pi/2, 3\pi/2]$.  
Here, for $0\le t\le \pi$, we have 
\begin{align*} r e(\theta) \in E(t) & \iff r \le 2\cos(\theta-t) \\
                      & \iff \theta \in [t-\arccos (r/2), t+\arccos (r/2)]. 
\end{align*}
Therefore, by setting $a:=\arccos (r/2) \in [0, \pi/2]$, we get
\begin{align*} r e(\theta) \in E(\pi)\cap E(t) & \iff \theta \in [\pi-a,\pi+a] \cap [t-a, t+a] \\
                                 & \iff \theta \in [\pi-a, t+a]. \end{align*}
Thus, the inequalities $0\le t'\le t \le \pi$ imply 
\[ E(\pi)\cap E(t') \subset E(\pi)\cap E(t). \]

\end{proof}

Proposition~\ref{prop:norm_ineq} follows immediately from the following:.
\begin{lem}
Let $n\in\Z_{\ge0}$. Assume that $(2n+1)$ real numbers $\theta_0,\dots,\theta_{2n}\in\R$ satisfy the inequalities 
$\theta_0 \le \dots \le \theta_{2n} \le \theta_0+\pi$. 
Then we obtain
\[ \sum_{i=0}^{2n}(-1)^i e(\theta_i) \in D(\theta_0,\theta_{2n}). \]
\end{lem}
\begin{proof}
We prove this by induction on $n$.
For $n=0$, taking any $\theta_0\in \R$, we have
\[ D(\theta_0,\theta_0)=\{e(\theta_0)\}, \]
which proves our lemma for $n=0$. 
Next, assuming that the lemma holds for $n-1$, we show that it holds for $n$.
Take $\theta_i\in\R$ with $\theta_0\le \dots \le \theta_{2n} \le \theta_0+\pi$ and put 
\begin{align*} x:=\sum_{i=0}^{2n}(-1)^i e(\theta_i), \quad 
   y:=\sum_{i=1}^{2n-1}(-1)^i e(\theta_i). 
\end{align*}
By the induction hypothesis and Lemma~\ref{note:dst}, $y\in D(\theta_1,\theta_{2n-1}) \subset D(\theta_0,\theta_{2n})$ holds.
Noting that $x=e(\theta_0) - y + e(\theta_{2n})$, we obtain
\begin{align*} y\in D(\theta_0,\ \theta_{2n}) & \iff \|y\|\le 1,\ \|y-(e(\theta_0)+e(\theta_{2n}))\|\le 1 \\
                               & \iff \|x-(e(\theta_0)+e(\theta_{2n}))\|\le 1,\ \|-x\|\le 1 \\
                               & \iff x\in D(\theta_0,\ \theta_{2n}).  \end{align*}
Thus, our lemma holds for $n$. 
\end{proof}

\subsection{The condition (d)}
Using Proposition~\ref{prop:norm_ineq} in the previous section, 
we verify the condition (d).

We show that for any $j\in\{1,\dots,N\}$, the norm of
\[ (Y\;\t\!B)_j = \frac12 \hat K \;\t\! C \;\t\!B e_j \]
is at most $1$. Namely, putting 
\[ s_j := \frac12 \t\!C \;\t\!B e_j, \]
we prove the norm of $\hat K s_j$ is at most $1$.
We note that $\hat K=(\hat k_1,\dots,\hat k_n)$. 
That is, viewing $\hat K$ as a collection of column vectors, they are all unit vectors, and $\hat k_1, \dots, \hat k_n$ are arranged in counterclockwise order. 
On the other hand, the following assertion on $s_j \in \R^n$ holds. 
\begin{lem}\label{lem:sum:binary} Let $i\in\{1,\dots,n\}$, $j\in\{1,\dots,N\}$, and $s_{ij}$ denotes the $i$-th component of $s_j$.
Then we obtain 
\begin{enumerate}
\item[\rm(i)] $\displaystyle \A k\in \{1,\dots,n\},\ \sum_{i=1}^k s_{ij} \in \{0,1\}$, 
\item[\rm(ii)] $\displaystyle \sum_{i=1}^n s_{ij} = 1$. 
\end{enumerate}
\end{lem}
Assuming this lemma, each component of $s_j$ is either $0$, $1$, or $-1$. Furthermore, ignoring the zeros and reading from top to bottom, the remaining entries alternate as $1, -1, 1, \dots, -1, 1$, starting and ending with $1$. Therefore, it follows from Proposition~\ref{prop:norm_ineq} that the norm of $\hat K s_j$ is at most $1$.

\begin{proof}[Proof of Lemma~\ref{lem:sum:binary}]
Let $\varepsilon_i$ ($i= 1, \ldots, n$) be the standard basis of $\mathbb{R}^n$.

Then, since 
\[
s_{ij} = {}^t\varepsilon_{i}s_{j} = \frac{1}{2} {}^t\varepsilon_{i}{}^tC{}^tBe_j,
\]
we obtain
\[
\sum_{i=1}^{k}s_{ij} = \frac{1}{2}{}^t\biggl(\sum_{i=1}^{k}C\varepsilon_i\biggr){}^tBe_j.
\]

By the definition of $B$ and $C$, we easily see that
\[
\sum_{i=1}^{k}C\varepsilon_i=
\begin{cases}
\varepsilon_1 - \varepsilon_{k+1} & (1\le k< n),\\
2\varepsilon_1  & (k=n)
\end{cases}
\]
and
\[
{}^tBe_j = \varepsilon_1 \pm \varepsilon_2 \pm \cdots\pm \varepsilon_n.
\]
The verification of the condition (d) now follows immediately.
\end{proof}


\begin{thebibliography}{9}
\bibitem{AK20} N.~Andrejic, R.~Kunjwal, Joint measurability structures realizable with qubit measurements: Incompatibility via marginal surgery, Phys. Rev. Res., \textbf{2}, 043147 (2020). 
\bibitem{B86} P.~Busch, Unsharp reality and joint measurements for spin observables, Phys. Rev. D \textbf{33}, 2253--61 (1986). 
\bibitem{GU25} D.~Grinko, R.~Uola, Compatibility of Binary Qubit Measurements, Phys. Rev. Let. \textbf{135} 200201 (2025).
\bibitem{R70} R.~T.~Rockafellar, Convex analysis, Princeton University Press (1970). 
\end{thebibliography}
\end{document}